\documentclass[prl,twocolumn,a4paper,nofootinbib,superscriptaddress]{revtex4}

\usepackage{hyperref}
\usepackage{graphicx}
\usepackage{amsmath}
\usepackage{amssymb}
\usepackage{amsthm,dsfont}
\usepackage{color}

\theoremstyle{plain}
\newtheorem{thm}{Theorem}

\newcommand{\lbold}{\boldsymbol{\lambda}}
\newcommand{\mc}[1]{\mathcal{#1}}
\newcommand{\mbb}[1]{\mathbb{#1}}
\newcommand{\tr}{\text{Tr}}
\newcommand{\abs}[1]{\left|#1\right|}
\newcommand{\pr}{{\rm Pr}}
\newcommand{\ct}{^{\dagger}}

\begin{document}

\title{Estimating outcome probabilities of quantum circuits using quasiprobabilities}

\author{Hakop Pashayan}
\affiliation{Centre for Engineered Quantum Systems, School of Physics, The University of Sydney, Sydney, NSW 2006, Australia}
\author{Joel J. Wallman}
\affiliation{Institute for Quantum Computing and Department of Applied Mathematics,
	University of Waterloo, Waterloo, Ontario, Canada, N2L 3G1
}
\author{Stephen D. Bartlett}
\affiliation{Centre for Engineered Quantum Systems, School of Physics, The University of Sydney, Sydney, NSW 2006, Australia}

\date{10 August 2015}

\begin{abstract}
We present a method for estimating the probabilities of outcomes of a quantum circuit using Monte Carlo sampling techniques applied to a quasiprobability representation. Our estimate converges to the true quantum probability at a rate determined by the total negativity in the circuit, using a measure of negativity based on the 1-norm of the quasiprobability. If the negativity grows at most polynomially in the size of the circuit, our estimator converges efficiently.  These results highlight the role of negativity as a measure of non-classical resources in quantum computation.
\end{abstract}

\maketitle

Estimating the probability of a measurement outcome in a quantum process using 
only classical methods is a longstanding problem that remains of acute interest
today. Directly calculating such probabilities using the Born rule is inherently
inefficient in the size of the quantum system, and efficiently estimating 
such probabilities for a generic quantum process is expected to be out of reach
of classical computers.  

Nonetheless, there are interesting and nontrivial classes of quantum circuits 
for which we \textit{can} efficiently estimate the probabilities of outcomes.  
The canonical example of such a class is that of stabilizer circuits. Such 
circuits can create highly-entangled states and perform many of the fundamental 
operations involved in quantum computing (teleportation, quantum error 
correction, distillation of magic states) but the celebrated Gottesman-Knill 
theorem allows such circuits to be classically simulated 
efficiently~\cite{AaronsonGottesman}. Other examples include fermionic linear optics/matchgates~\cite{Valiant2002,Divincenzo2002}, and some classes of quantum 
optics~\cite{Bartlett2002,Veitch2013b,Gurvits}. While these methods may be extended to 
include bounded numbers of operations outside of the class (for example, 
Ref.~\cite{AaronsonGottesman}), such extensions generally treat all 
operations outside of the class on an equal footing (for example, the cost of 
adding noisy magic states is the same as adding pure magic states) and so do 
not provide any insight into the relative resources of different operations.

In this Letter, we present a general method for estimating outcome probabilities for quantum circuits using quasiprobability representations. Simulation methods based on quasiprobability representations have a long history in physics~\cite{Gardiner2004}, and have recently been used in quantum computation to identify classes of operations that are efficiently simulatable~\cite{Veitch2012,Mari,Stahlke2014}. Our method allows for estimation in circuits wherein the quasiprobabilities may go negative. That is, while making the most efficient use of circuit elements that are represented nonnegatively, it nonetheless provides an unbiased estimator of the true quantum outcome probability regardless of the inclusion of more general elements that are negatively represented. We quantify the performance of this method by providing an upper bound on the rate of convergence of this estimator that scales with a measure of the total amount of negativity in the circuit.

\textit{Probability estimation.---}Consider quantum circuits of the following 
form. The circuit initiates with $N$ qudits ($d$-level quantum systems) in a 
product state, evolves through a circuit consisting of $L=poly(N)$ elementary 
gates that act nontrivially on at most a fixed number of qudits (for 
example, 1- and 2-qudit gates), and terminates with a product measurement, 
i.e., an independent measurement of each qudit. Universal quantum computation 
can be achieved with circuits of this form. Note that we do not include 
circuits with intermediate measurements and conditional operations based on 
their outputs (we return to this consideration in the discussion).

We aim to estimate the probability of a fixed outcome 
$\vec{o}=(o_1,\ldots,o_N)$ where $o_j$ denotes the outcome of the measurement 
on the $j$th qudit. (Note that \textit{estimation} of the probability of a 
fixed outcome is distinct from a \textit{simulation} as in 
Refs.~\cite{Veitch2012, Mari}, wherein different outcomes are sampled from this 
distribution.)  A natural benchmark for the precision of an estimator is the precision that 
can be obtained from sampling the quantum circuit itself. If we had
access to a quantum computer that implemented a circuit in this class, then we
could use it to estimate the probability of a fixed outcome by computing the
observed frequency $f_s(\vec{o})$ of outcome $\vec{o}$ over $s$ samples. By the Hoeffding inequality,
$f_s(\vec{o})$ will be within $\epsilon$ of the quantum probability $p(\vec{o})$
with probability $1-\delta$ provided the number of samples
$s(\epsilon,\delta)$ satisfies
\begin{equation}
s(\epsilon,\delta) \geq \tfrac{1}{2 \epsilon^2}\log{(2/\delta)}.
\end{equation}
This bound implies that for any fixed $\delta$, the number of samples required 
to achieve $\epsilon$ error scales polynomially in $1/\epsilon$. We call 
estimators satisfying this property \textit{poly-precision} estimators. (We 
distinguish these from \textit{exponential-precision} estimators, defined as 
estimators for which $s(\epsilon,\delta)$ scales logarithmically in 
$1/\epsilon$.)

Our central results are a classical algorithm that produces a poly precision estimate of a quantum circuit in the above class, and a bound on the efficiency of this algorithm based on a measure of the circuit's negativity in a quasiprobability representation.  

\textit{Quasiprobability representations.---}A quasiprobability representation of a qudit over $\Lambda$ is defined~\cite{Ferrie2009,Ferrie2011} by a frame $\{F(\lambda):\lambda\in\Lambda\}$ and 
a dual frame $\{G(\lambda):\lambda\in\Lambda\}$, which are (generally 
over-complete) bases for the space of 
Hermitian operators acting on $\mbb{C}^d$ satisfying 
$A = \sum_{\lambda\in \Lambda} G(\lambda)\tr[AF(\lambda)]$ for all  
$A$.  The space $\Lambda$ can be continuous or discrete, and although many quasiprobability representations assume a phase space (symplectic) structure on $\Lambda$, this is not necessary.
We can define quasiprobability distributions on $\Lambda$ associated with a quantum state $\rho$, a 
unitary operator $U$ and a measurement effect $E$ to be
\begin{align}
W_{\rho}(\lambda) &= \tr[F(\lambda)\rho] \,,	\notag\\
W_{U}(\lambda'|\lambda) &= \tr(F(\lambda') U G(\lambda) U^\dag)\,,	\notag\\
W(E|\lambda) &= \tr[E G(\lambda)]	\,.
\label{eq:quasidist}
\end{align}
Tensor products of these dual frames gives a dual frame for the product 
space, and so these definitions 
extend in the obvious way from a tensor product $(\mbb{C}^d)^{\otimes N}$ of $N$ qudits 
to distributions on a phase space $\Lambda^{N}$.

The distribution $W_{\rho}(\lambda)$ is real-valued and satisfies $\sum_{\lambda \in \Lambda} W_{\rho}(\lambda)=1$, much like a probability distribution, if the frame is normalized using $\sum_{\lambda \in \Lambda} F(\lambda) = I$.  Similarly, the distributions $W_{U}(\lambda'|\lambda)$ and $W(E|\lambda)$ are normalized like corresponding conditional probabilities.  The Born rule $\mathrm{Pr}(E|\rho, U) = \tr(E U\rho U^\dag)$, which gives the quantum probability for a measurement outcome given the state and process, is reproduced in the quasiprobability representation as would be expected in a probabilistic theory, by
\begin{equation}\label{eq:ontic_probability}
\mathrm{Pr}(E|\rho, U) = \sum_{\lambda,\lambda' \in \Lambda} W(E|\lambda')
W_{U}(\lambda'|\lambda)W_{\rho}(\lambda)	\,.
\end{equation}
This equation follows from the Born rule using the definition of the dual frames.  

Importantly, the distributions of a quasiprobability representation will generally take on negative values, and so cannot be directly interpreted as probability distributions. The 1-norm of a quasiprobability distribution provides a natural measure of the amount of negativity, i.e., how much it deviates from a true probability distribution.  We define the \textit{negativity} $\mc{M}_\rho$ of a state $\rho$ as the 1-norm of its quasiprobability representation, 
\begin{equation}
\label{eq:mana}
	\mc{M}_{\rho}= ||W_\rho ||_1 = \sum_{\lambda\in\Lambda}\abs{ W_{\rho}(\lambda)}\,.
\end{equation}
(The \textit{mana} of a state using the discrete Wigner representation was introduced in Ref.~\cite{Veitch2013} as a measure to bound the resources required for magic state distillation, and defined as the logarithm of the negativity used here.) Analogously, we define the \textit{negativity} $\mathcal{M}_E$ of a measurement effect $E$ to be
\begin{equation}
\mc{M}_E =\sum_{\lambda\in\Lambda}\abs{ W(E|\lambda)},
\end{equation}
and the \textit{point-negativity} $\mc{M}_U (\lambda)$ and \textit{negativity} $\mc{M}_U$ of a unitary $U$ to be
\begin{equation}
\mc{M}_U(\lambda) =  \sum_{\lambda'\in\Lambda}\abs{ W_U(\lambda'|\lambda)} , \quad
\mc{M}_U = \max_{\lambda\in \Lambda} \mc{M}_U(\lambda)	\,,
\end{equation}
respectively. The negativities for states, unitaries and effects are lower-bounded by 1, 1 and $\tr(E)$ respectively, with equality if and only if the quasiprobability representation is nonnegative. These negativities will serve as a measure of the cost of each circuit element in our estimator. 

\textit{Estimation procedure.---}A quasiprobability representation provides an interpretation of the Born rule as the expectation value of a stochastic process.  Specifically, if the quasiprobability representation of all elements in the circuit are nonnegative, then one may interpret the Born rule of Eq.~\eqref{eq:ontic_probability} as the expected probability of the measurement outcome averaged over a set of trajectories through phase space~\cite{Veitch2012,Mari}.  
This stochastic interpretation no longer holds if the quasi-probability representation for any of the input states, gates, or measurements is negative.  A standard perspective is that, in a quantum description, different trajectories in phase space can be assigned negative weights and can interfere with each other~\cite{Stahlke2014}. 
Monte Carlo sampling techniques may still be used, but the key problem is to identify an appropriate distribution to sample trajectories $\lbold =(\lambda_0,\ldots,\lambda_L)$ through phase space, where $\lambda_0$ is associated with the preparation and $\lambda_{l}$ and $\lambda_{l-1}$ are associated with the $l$th unitary.  Eq.~\eqref{eq:ontic_probability} becomes $\pr(E|\rho,U) = \sum_{\lbold} W(\lbold)$ with
\begin{equation}\label{eq:ontic_probability_chain}
W(\lbold) = W(E|\lambda_L)\bigl[\textstyle\prod_{l=1}^L 
W_U(\lambda_l|\lambda_{l-1})\bigr] W_{\rho}(\lambda_0) \,.
\end{equation}
Using an approach reminiscent of quantum Monte Carlo methods for fermion systems, we could sample trajectories from a true (nonnegative) probability distribution obtained from the absolute value of the quasiprobability, keeping track of the sign of the sampled trajectory.  Consider the distribution of trajectories given by 
\begin{align}\label{eq:LVUED}
\pr(A=\lbold) = 
\frac{\abs{W(\lbold)}}{\mc{M}_{\rm c}}
\end{align}
where $\mc{M}_{\rm c} = \sum_{\lbold}\abs{W(\lbold)}$ measures the negativity 
of the entire circuit, and 
we regard $\lbold$ as a realization of a random variable 
$A$. An estimate based on a single realization $\lbold$ is given by $\hat{q}_1(\lbold) = \mathcal{M}_{\rm c}\, {\rm Sign}[W(\lbold)]$, where ${\rm Sign}[\cdot] = \pm 1$ depending on the sign of the input.  The expected value of this estimate gives the desired Born rule probability
\begin{multline}
\langle \hat{q}_1(A) \rangle = \sum_{\lbold} \hat{q}_1(\lbold) \pr(A=\lbold)	= \sum_{\lbold} {\rm Sign}[W(\lbold)] |W(\lbold)| \\
  = \sum_{\lbold} W(\lbold) = \pr(E|\rho,U)\,. \label{eq:unbiased}
\end{multline}
Note that this estimator minimizes the range of $A$~\cite{Stahlke2014}, and so provides the best bound on the number of samples required to obtain a fixed precision using the Hoeffding inequality.  Sampling from the distribution~\eqref{eq:LVUED} is also the optimal estimator over the space of trajectories in that it has the smallest variance (see Appendix).
Unfortunately, without any additional structure, this estimator will in general be impractical for two reasons: there is no known efficient method to compute $\mc{M}_{\rm c}$, and sampling from the distribution~\eqref{eq:LVUED} will in general be inefficient in $N$.

To develop an efficient procedure, we sample trajectories $\lbold$ following a Markov chain, using (true) probabilities and conditional probabilities at each timestep.  Consider an input product state $\rho = \otimes_{n=1}^N \rho_n$, which has an efficient description $W_{\rho}(\lambda)$ in the quasi-probability representation which may be negative. We sample the initial point of the trajectory $\lambda_0$ from the modified distribution $\pr(\lambda_0)=\abs{W_{\rho}(\lambda_{0})}/\mc{M}_{\rho}$.  We construct a full trajectory $\lbold$ by sampling $\lambda_l$ at each timestep $l=1,\ldots,L$ in the circuit from the conditional distribution $\pr(\lambda_l|\lambda_{l-1}) = \abs{W_{U_l}(\lambda_l|\lambda_{l-1})}/\mc{M}_{U_l}(\lambda_{l-1})$ given by the unitary gate $U_l$.  If the unitary $U_l$ at each timestep $l$ of the circuit has an efficient description in the quasi-probability representation, for example it consists of quantum gates acting on a fixed number of systems, then this distribution can be sampled efficiently.  We note that this efficiency comes at a cost, as trajectories are no longer sampled from the optimal distribution~\eqref{eq:LVUED}.  

An estimate $\hat{p}_1$ based on a single trajectory $\lbold$ of our Markov chain protocol is given by 
\begin{multline}
\hat{p}_1(\lbold) =\mc{M}_{\rho}{\rm Sign}  [W_{\rho}(\lambda_{0})] \\ \times
\textstyle\prod_{l=1}^L \Bigl[ \mc{M}_{U_l}(\lambda_{l-1}){\rm Sign}[W_{U_l}(\lambda_l|\lambda_{l-1})] \Bigr]  W(E|\lambda_{L})	\,.
\end{multline}
Unlike $\hat{q}$, this estimate is guaranteed to be an efficiently computable function of the sampled path $\lbold$. We note that $\hat{p}_1$ can lie outside the unit interval, but nonetheless gives an unbiased estimate of the Born rule probability, $\langle \hat{p}_1(A) \rangle = \pr(E|\rho,U)$, precisely as in Eq.~\eqref{eq:unbiased}.  Further, we note that $\hat{p}_1$ lies in the interval 
$[-\mc{M}_\rightarrow,+\mc{M}_\rightarrow]$, where we have defined 
$\mc{M}_\rightarrow$ to be the \textit{total forward negativity bound} of the circuit:
\begin{equation}
\label{eq:forward}
\mc{M}_{\rightarrow} = \mc{M}_{\rho}\textstyle\prod_{l=1}^L\mc{M}_{U_l}\max_{\lambda_L}\abs{W(E|\lambda_L)}	\,.
\end{equation}
Let $\hat{p}_s$ be the average of $\hat{p}_1$ over $s$ independent samples of $\lbold$.
Using the boundedness and unbiasedness properties of $\hat{p}_1$, the Hoeffding inequality yields an upper bound on the rate of convergence of the average $\hat{p}_s$.   Specifically, $\hat{p}_s$ will be within $\epsilon$ of the quantum probability ${\rm Pr}(E|\rho,U)$ with probability $1-\delta$ if a total of 
\begin{equation}\label{eq:convergence}
  s(\epsilon,\delta)  = \tfrac{2}{\epsilon^2} \mc{M}_\rightarrow^2 \ln(2/\delta)
\end{equation}
samples are taken. Consequently, if the total forward negativity bound $\mc{M}_\rightarrow$ grows at most polynomially with $N$, then our protocol gives an efficient estimate $\hat{p}_s$ of the quantum probability $\pr(E|\rho,U)$ to within $\epsilon = 1/poly(N)$, with an exponentially small failure probability. That is, for circuits with a polynomially-bounded total forward negativity bound, $\hat{p}_s$ is a poly-precision estimator of the Born rule probability and we can sample $\hat{p}_s$ efficiently in $N$.  We note that the total forward negativity bound $\mathcal{M}_\rightarrow$ of~\eqref{eq:forward} is insensitive to the measurement negativity $\mathcal{M}_E$, instead depending only on $\max_{\lambda_L}\abs{W(E|\lambda_L)}$. 

Any efficiently computable symmetry of the Born rule can be used to give a variant on the procedure defined above.  The rate of convergence of the estimator need not be symmetric under these Born rule symmetries, and so such a variant may provide an advantage.  Two examples of such symmetries -- the time reversal symmetry that exchanges states and measurement effects in a unitary circuit, and the regrouping of unitaries into different elementary gates -- are explored in the Appendix.  In particular, a variant procedure is presented for which the total negativity bound is insensitive to the negativity of the initial state $\mathcal{M}_\rho$.

\textit{Example: Estimation with the discrete Wigner function.---}The odd-$d$ qudit stabilizer subtheory and the associated discrete Wigner function provide a 
canonical example for demonstrating the use of our algorithm; see also Ref.~\cite{Stahlke2014}.
Using this discrete Wigner representation for our estimation algorithm, the 
nonnegativity of the stabilizer subtheory~\cite{Cormick2006,Gross2006} ensures that stabilizer states, 
gates, and rank-1 measurements have negativity $\mc{M}_{\rho_n/U_l/E_n}=1$ and so are 
``free'' resources. Moreover, due to the existence of nonnegatively represented operations that are not in the stabilizer polytope~\cite{Veitch2012}, our approach is efficient on a strictly larger set of circuits than those of Ref.~\cite{AaronsonGottesman}.  Circuits with operations possessing negativity strictly greater than 1, such as magic states and non-Clifford gates, can still be estimated but now come at a cost.  Provided the total negativity bound grows at most polynomially in $N$, our protocol provides an efficient estimator.

As an example, consider a circuit with an input 
state given by a product state of $k$ qutrit magic states 
$\tfrac{1}{\sqrt{3}}(|0\rangle + \xi|1\rangle + \xi^8|2\rangle)$, with $\xi=\exp(2\pi i/9)$, together with 
stabilizer $|0\rangle$ states in a 100-qutrit random Clifford circuit, and 
estimate the probability of measuring $|0\rangle$ on the first qutrit of the 
output.  The total forward negativity bound of this circuit scales exponentially in $k$ and 
consequently the number of samples required to guarantee a fixed precision 
scales exponentially in $k$ by Eq.~\eqref{eq:convergence}. The results of our 
numerical simulations, shown in Fig.~\ref{fig:RandomCliffords}, indicate that 
our estimator does indeed converge with an appropriately chosen number of 
samples. Moreover, while the true precision of $\hat{p}_s$ in our simulations 
is often orders of magnitude better than the target precision, there are 
circuits that come close to saturating the target precision, suggesting that 
our bound cannot be substantially improved without further detailed knowledge 
of the circuit.

\begin{figure}
\centering
\includegraphics[width=0.95\linewidth]{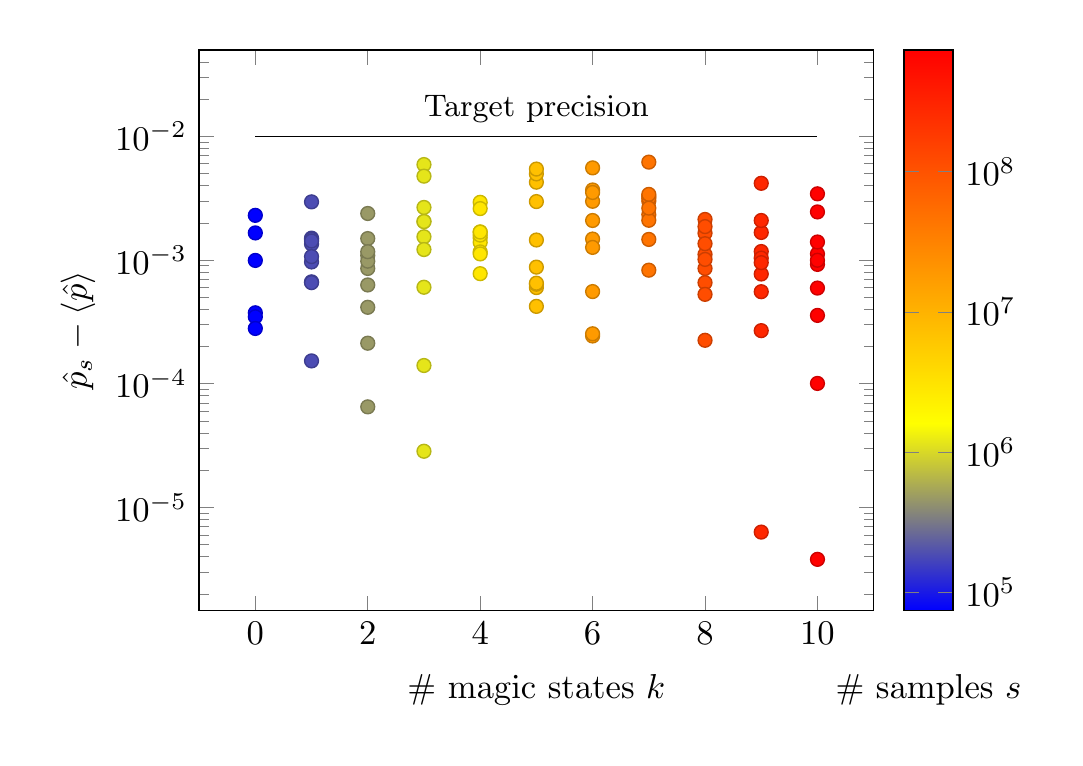}
\caption{(Color online) Plot of the difference $\hat{p}_s- 
\langle\hat{p}\rangle$ between the estimated probability and the true 
probability of the outcome $|0\rangle\langle 0|$ for the first qutrit as a function 
of the number of magic states $k$. Each data point represents a random 
100-qutrit Clifford circuit with the non-magic states initialized to the 
$|0\rangle$ state. The number of samples $s(k)$ was chosen using 
Eq.~\eqref{eq:convergence} with target precision $\epsilon = 0.01$ (indicated 
by the solid line) with 95\% confidence ($\delta =0.05$), so the number of 
samples increases exponentially with $k$ (color scale).
}\label{fig:RandomCliffords}
\end{figure}

\textit{Discussion.---}Our results highlight the role of the total negativity of a circuit as a resource required for a 
quantum computer to outperform any classical computer. In particular, any 
circuit element that is represented nonnegatively does not contribute to the total negativity bound and can be viewed as a ``free'' resource within the 
algorithm.  Other circuit elements have an associated cost quantified by their 
negativity, \textit{unless} they appear at the final timestep of the algorithm.  This latter observation motivates us to exploit the time-reversal and other symmetries of the Born rule, seeking to minimize the total forward (or reverse) negativity bound.  In particular, one could seek equivalent circuits wherein negative operations can be replaced with nonnegative ones by using negative initial states or measurements via gate teleportation.  By choosing the forward or reverse procedure as appropriate, the efficiency can be made insensitive to the negativity of these initial states or measurements.  

It also motivates us to identify quasiprobability representations in which 
many of the circuit elements of interest are represented nonnegatively.  
Interesting and relevant examples abound, beyond the well-studied qudit discrete Wigner 
function.  Discrete Wigner functions for qubits ($d=2$) can be defined for 
which all stabilizer states with real coefficients (rebits) and all 
CSS-preserving unitaries are nonnegatively represented~\cite{Delfosse2014}.  
The range of quasiprobability representations introduced in 
Ref.~\cite{Wallman2013} represent discrete subgroups of $U(2)$ on a single qubit nonnegatively, but have no nonnegative entangling gates; as
such representations can represent certain non-stabilizer single-qubit states nonnegatively, they may be useful for estimation in circuits for gate 
synthesis.  
There is also flexibility in how a quasiprobability representation is defined.  For example, a given quasiprobability representation can be modified to describe more 
states nonnegatively at the expense of a decreasing set of nonnegative 
measurements, and vice versa, by exploiting the structure of the dual frames.  Overcomplete frames provides freedom in the choice of dual frame, and the negativity of unitaries and measurements will depend on this choice.  As the dual frames formalism itself captures the relationship between quantum states and measurements, there is also freedom in the definition of the quasiprobability representation of unitaries beyond that given by Eq.~\eqref{eq:quasidist}.  Finally, again using the freedom in the choice of dual for overcomplete frames, it is possible to switch between frames throughout a single circuit.  These freedoms can be used to minimize the total negativity bound of the circuit, allowing more efficient estimators.

Our procedure can be applied to infinite-dimensional Hilbert spaces using any of the range of quasiprobability representations with continuous phase spaces developed in the study of quantum optics, by performing an appropriate discretization as in Ref.~\cite{Veitch2013b}.  In this case, the negativity of distributions is quantified by integrating the absolute value of the distributions over the phase space, and is directly related to the \textit{volume} of negativity~\cite{Kenfack2004}.  We note that the resulting estimator can be applied to quantum optics experiments including states and measurements with negative Wigner function, such as photon number Fock states, and so may provide additional insight into the classical simulation cost of boson sampling~\cite{Aaronson2011}.  While there exist means to efficiently estimate the outcome probability of a specific linear optics circuit with Fock state input and measurement~\cite{Gurvits}, our estimation procedure extends these results by providing a general method for estimating outcome probabilities of such linear optical circuits for any input and output together with a bound on the efficiency of this estimation based on the volume of negativity of these states.  In addition, our estimation can easily incorporate squeezing, as well as the loss and noise mechanisms common to linear optics experiments.

There are two natural ways to extend our results to circuits that include intermediate measurements and conditional operations based on them.  First, one could replace the measurement and conditional operation with a coherently-controlled operation, and delay the measurement to the end.  We note that such controlled operations can be negative, even if the measurement and classically-controlled operation are both nonnegative.  Second, our algorithm can be used to directly estimate the probabilities of the intermediate measurements and to sample from them.  In this case, the required precision is exponential in the number of intermediate measurements in order to calculate conditional probabilities for subsequent use in the algorithm.  Thus, in general, both approaches require resources that are exponential in the number of intermediate measurements. 

Finally, our estimation procedure provides insight into the study of operationally meaningful measures of non-classical resources in quantum computation.  Negativity in a quasiprobability representation has long been used as an indicator of quantum behaviour, but only recently has it been quantified as a resource for quantum computation~\cite{Veitch2013}.  Our results provide a related operational meaning of this resource:  as a measure that bounds the efficiency of a classical estimation of probabilities.  

\begin{acknowledgments}
The authors are grateful to J.~Emerson, C.~Ferrie, S.~Flammia, R.~Jozsa and A.~Krishna for helpful discussions. This work is supported by the ARC via the Centre of Excellence in Engineered Quantum Systems (EQuS) project number CE110001013 and by the U.S. Army Research Office through grant W911NF-14-1-0103.
\end{acknowledgments}

\newpage

\appendix

\section{Proof of optimality}

In this section of the Supplementary Information, we prove that sampling the distribution
\begin{align}\label{eq:LVUED}
\pr(A=\lbold) = 
\frac{\abs{W(\lbold)}}{\mc{M}_{\rm c}}
\end{align}
given in Eq.~(8) is optimal in that it gives the minimum variance of the 
Born rule estimator of all distributions. 

Our goal is to estimate $P \equiv \pr(E|\rho,U) = \sum_{\lbold} W(\lbold)$.  We will do so by sampling $\lbold$ from some distribution $p(\lbold)$, and compute an estimator given by $W(\lbold)/p(\lbold)$.  The variance $V$ of this estimator is given by
\begin{equation}
	V = \sum_{\lbold} \frac{[W(\lbold)]^2}{p(\lbold)} - P^2 \,.
\end{equation}
To minimize the variance, we choose $p(\lbold)$ to minimize the first term in this expression.  

\begin{thm}
The distribution $p(\lbold)$ that minimizes the variance is $p(\lbold) \propto \abs{W(\lbold)}$.
\end{thm}

\begin{proof}
The proof follows directly from the Cauchy-Schwarz inequality.  Consider
\begin{align}
	\left( \sum_{\lbold} \abs{W(\lbold)} \right)^2 &= \left( \sum_{\lbold} \frac{\abs{W(\lbold)}}{\sqrt{p(\lbold)}} \sqrt{p(\lbold)} \right)^2 \notag \\
	&\leq \left( \sum_{\lbold} \frac{[W(\lbold)]^2}{p(\lbold)} \right) \left( \sum_{\lbold} p(\lbold) \right) \notag \\
	&= \sum_{\lbold} \frac{[W(\lbold)]^2}{p(\lbold)}
\end{align}
The inequality is saturated for $p(\lbold) \propto |W(\lbold)|$, and therefore this distribution minimizes the variance $V$. 
\end{proof}

With this choice, the variance of the estimator is
\begin{equation}
	V_\text{min} = \left( \sum_{\lbold} \frac{\abs{W(\lbold)}}{\mc{M}_{\rm c}} 
	\right)^2 - P^2 \,.
\end{equation}
where $\mc{M}_{\rm c} = \sum_{\lbold}\abs{W(\lbold)}$.

\section{Exploiting symmetries of Born rule}

In this section of the Supplementary Information, we detail ways to exploit an efficiently computable 
symmetry of the Born rule to give a variant of our estimation algorithm.  If we replace the quantum circuit and measurement effect with another such that the Born rule probability remains the same, then Eq.~(10) provides two (in general) different estimators for this Born rule probability.  The rate of convergence of these estimator need not be 
the same under this symmetry, and so such a variant may provide an 
advantage.  

As an example, consider the `time reversal' symmetry that exchanges 
states and measurement effects in a unitary circuit (with some care  taken to 
appropriately normalise the distributions for states and effects).  One can 
define  a ``reverse protocol'' which produces a poly precision estimator 
$\hat{p}'_s$, provided that the \textit{total reverse negativity} of the circuit 
\begin{equation}
\mc{M}_{\leftarrow} = \mc{M}_{E}\textstyle\prod_{l=1}^L\mc{M}_{U_l\ct}\max_{\lambda_0}\abs{W_{\rho}(\lambda_0)}	\,,
\end{equation}
is polynomially bounded.  In general, $\mc{M}_{\rightarrow} \neq \mc{M}_{\leftarrow}$, as seen from 
\begin{equation}
  \label{eq:ratio}
\frac{\mc{M}_{\leftarrow}}{\mc{M}_{\rightarrow}}=\frac{\mc{M}_{E}}{\mc{M}_{\rho}}
\frac{\max_{\lambda_0}\abs{W_{\rho}(\lambda_0)}}{\max_{\lambda_L}\abs{W(E|\lambda_L)}}
\frac{\prod_{l=1}^L\mc{M}_{U_l\ct}}{\prod_{l=1}^L\mc{M}_{U_l}}
\,.
\end{equation}
Because both $\mc{M}_{\rightarrow}$ and $\mc{M}_{\leftarrow}$ are efficiently 
computable, one is free to choose the direction of simulation resulting in the 
faster estimator convergence rate.  (We note that $\mc{M}_{E} \geq \tr(E)$ 
while $\mc{M}_{\rho} \geq 1$, which suggests that the reverse protocol would 
have slower convergence when using a high rank effect. However, in such cases, 
${\max_{\lambda_L}\abs{W(E|\lambda_L)}}$ is in general larger 
than ${\max_{\lambda_0}\abs{W_{\rho}(\lambda_0)}}$ by a similar factor, cancelling the effect of $\tr(E)$ in the ratio \eqref{eq:ratio}.)

Another symmetry of the Born rule is the the regrouping of unitaries into different `elementary' gates, such as reexpressing $U = U_L U_{L-1}\cdots U_1$ as $U= U'_{L'}U'_{L'-1}\cdots U'_1$.  Different groupings can lead to different estimators, as we demonstrate with a simple example using a grouping of two unitaries into one, $U=U_2 U_1$.  We can estimate $p={\rm Tr}(U\rho U^{\dag}E)$ by sampling trajectories $\lbold=(\lambda_0,\lambda_1,\lambda_2)$ using $U_1$ and $U_2$, or directly by sampling $\lbold'=(\lambda_0,\lambda_2)$ using $U=U_2 U_1$ as a single step.  While both of these methods will produce an unbiased estimator of the Born rule, they will not converge at the same rate in general, as a result of the general inequality
\begin{equation}
\sum_{\lambda_1 \in \Lambda} \left[\frac{\abs{W_{U_2}(\lambda_2|\lambda_1)}}{\mc{M}_{U_2}(\lambda_1)} \frac{\abs{W_{U_1}(\lambda_1|\lambda_0)}}{\mc{M}_{U_1}(\lambda_0)} \right] \neq \frac{\abs{W_{U}(\lambda_2|\lambda_0)}}{\mc{M}_{U}(\lambda_0)}	\,.
\end{equation}
We note that equality holds in the case where $U_1$ and $U_2$ are both nonnegative.

\end{document}